\documentclass{llncs}

\usepackage{amsmath,amssymb}

\usepackage{float}
\usepackage[dvips]{epsfig}
\usepackage{epsf}
\usepackage{url}
\usepackage{amscd}
\usepackage{graphicx}
\usepackage{enumerate}

\def\amod#1 #2{#1\ ({\rm mod}\ #2)}
\def\Enn{{\mathbb{N}}}
\newcommand{\nabs}[1]{\lvert #1\rvert}
\newcommand{\nfloor}[1]{\lfloor #1\rfloor}
\newcommand{\bsset}[1]{\bigl\{#1\bigr\}}
\newcommand{\ntoinf}{n \rightarrow \infty}
\newcommand{\bparen}[1]{\bigl(#1\bigr)}

\frontmatter

\title{Primitive Words and Lyndon Words in Automatic and Linearly Recurrent
Sequences}

\author{Daniel Go\v{c}\inst{1} \and Kalle Saari\inst{2} and Jeffrey Shallit\inst{1}}

\institute{School of Computer Science,
University of Waterloo,
Waterloo, ON  N2L 3G1,
Canada \\
\email{\{dgoc,shallit\}@cs.uwaterloo.ca} \\
\and
Mathematics and Statistics,
University of Winnipeg,
515 Portage Avenue,
Winnipeg, MB  R3B 2E9,
Canada \\
\email{kasaar2@gmail.com}  \\
}

\begin{document}

\maketitle

\begin{abstract}
We investigate questions related to the presence of primitive words and Lyndon words in automatic and linearly recurrent sequences.  We show that the Lyndon factorization of a $k$-automatic sequence is itself $k$-automatic.  We also show that the function counting the number of primitive factors (resp., Lyndon factors) of length $n$ in a $k$-automatic sequence is $k$-regular.  Finally, we show that the number of Lyndon factors of a linearly recurrent sequence is bounded.
\end{abstract}

\section{Introduction}

We start with some basic definitions.
A nonempty word $w$ is called a {\it power} if it can be written in the
form $w = x^k$, for some integer $k \geq 2$.  
Otherwise $w$ is called {\it primitive}.
Thus {\tt murmur} is a power,
but {\tt murder} is primitive.  A word $y$ is a {\it factor} of a word
$w$ if there exist words $x, z$ such that $w = xyz$.  If further
$x = \epsilon$ (resp., $z = \epsilon$), then $y$ is a {\it prefix} 
(resp., {\it suffix}) of $w$.  A prefix or suffix of a word $w$ is
called {\it proper} if it is unequal to $w$.

Let $\Sigma$ be an ordered alphabet.
We recall the usual definition of lexicographic order on the words
in $\Sigma^*$.
We write $w < x$ if either
\begin{itemize}
\item[(a)] $w$ is a proper prefix of $x$; or
\item[(b)] there exist words $y, z, z'$ and letters $a < b$ such that
$w = yaz$ and $x = ybz'$.
\end{itemize}
For example, using the usual ordering
of the alphabet, we have
${\tt common} < {\tt con} < {\tt conjugate}$.
As usual, we write $w \leq x$ if
$w < x$ or $w = x$.  

A word $w$ is a {\it conjugate} of a word $x$ if there exist words
$u, v$ such that $w = uv$ and $w = vu$.  Thus, for example,
{\tt enlist} and {\tt listen} are conjugates.  A word is said to
be {\it Lyndon} if it is primitive and lexicographically least among
all its conjugates.  Thus, for example, {\tt academy} is Lyndon,
while {\tt googol} and {\tt googoo} are not.    A classical theorem is
that a finite word is Lyndon if and only if it is lexicographically
less than each of its proper suffixes \cite{Duval:1983}.

We now turn to (right-) infinite words.  We write an infinite word in boldface,
as ${\bf x} = a_0 a_1 a_2 \cdots$ and use indexing starting at $0$.
For $i \leq j+1$, we let $[i..j]$ denote the set
$\lbrace i, i+1, \ldots, j \rbrace$.  (If $i = j+1$ we get the empty set.)
We let ${\bf x}[i..j]$ denote the word $a_i a_{i+1} \cdots a_j$.
Similarly, $[i..\infty]$ denotes the infinite set
$\lbrace i, i+1, \ldots \rbrace$ and
${\bf x}[i..\infty]$ denotes the infinite word $a_i a_{i+1} \cdots $.

An infinite word or sequence ${\bf x} = a_0 a_1 a_2 \cdots$ is said
to be {\it $k$-automatic} if there is a deterministic finite automaton
(with outputs associated with the states) that, on input $n$ expressed
in base $k$, reaches a state $q$ with output $\tau(q)$ equal to $a_n$.
For more details, see \cite{Cobham:1972} or 
\cite{Allouche&Shallit:2003}.  
In several previous papers 
\cite{Allouche&Rampersad&Shallit:2009,Charlier&Rampersad&Shallit:2011,Schaeffer&Shallit:2012,Shallit:2011,Goc&Henshall&Shallit:2012}, we
have developed a technique to show that many properties of
automatic sequences are decidable.    The fundamental tool is the
following:

\begin{theorem}
Let $P(n)$ be a predicate associated with a $k$-automatic sequence $\bf x$,
expressible using addition, subtraction, comparisons, logical
operations, indexing into $\bf x$, and existential and universal
quantifiers.  Then there is a computable
finite automaton accepting the base-$k$
representations of those $n$ for which $P(n)$ holds.  Furthermore, we can
decide if $P(n)$ holds for at least one $n$, or for all $n$, or for
infinitely many $n$.
\end{theorem}

If a predicate is constructed as in the previous theorem, we just
say it is ``expressible''.  Any expressible predicate is decidable.
As an example, we prove

\begin{theorem}
Let $\bf x$ be a $k$-automatic sequence.
The predicate $P(i,j)$ defined by ``${\bf x}[i..j]$ is primitive'' 
is expressible.
\label{prim-ex}
\end{theorem}

\begin{proof}  
(due to Luke Schaeffer)
It is easy to see that a word is a power if and only if it is equal to
some cyclic shift of itself, other than the trivial shift.
Thus a word is a power if and only if there is a $d$, 
$0 < d < j-i+1 $, such that $x[i..j-d] = x[i+d..j]$ and
$x[j-d+1..j] = x[i..i+d-1]$.   A word is primitive if there is
no such $d$.
\end{proof}

\begin{theorem}
Let $\bf x$ be a $k$-automatic sequence.
The predicate $LL(i,j,m,n)$ defined by ``${\bf x}[i..j] < {\bf x}[m..n]$''
is expressible.
\end{theorem}

\begin{proof}
We have ${\bf x}[i..j] < {\bf x}[m..n]$ if and only if
either
\begin{itemize}
\item[(a)] $j-i < n-m$ and ${\bf x}[i..j] = {\bf x}[m..m+j-i]$; or
\item[(b)] there exists $t < \min(j-i,n-m)$ such that
${\bf x}[i..i+t] = {\bf x}[m..m+t]$ and
${\bf x}[i+t+1] < {\bf x}[m+t+1]$.
\end{itemize}
\end{proof}

\begin{theorem}
Let $\bf x$ be a $k$-automatic sequence.
The predicate $L(i,j)$ defined by ``${\bf x}[i..j]$ is a Lyndon word''
is expressible.
\end{theorem}

\begin{proof}
It suffices to check that ${\bf x}[i..j]$ is lexicographically less
than each of its proper suffixes, that is, that 
$LL(i,j,i',j)$ holds for all $i'$ with  $i < i' \leq j$.
\end{proof}

We can extend the definition of lexicographic order to infinite
words in the obvious way.
We can extend the definition of Lyndon words to (right-) infinite words
as follows:  an infinite word ${\bf x} = a_0 a_1 a_2 \cdots$ is Lyndon 
if it is lexicographically less than all its suffixes
${\bf x}[j..\infty] = a_j a_{j+1} \cdots$ for $j \geq 1$.
Then we have the following theorems.

\begin{theorem}
Let ${\bf x}$ be a $k$-automatic sequence.
The predicate $LL_\infty(i,j)$ defined by
``${\bf x}[i..\infty] < {\bf x}[j..\infty]$ is expressible.
\end{theorem}

\begin{proof}
This is equivalent to  $\exists t \geq 0$ such that
${\bf x}[i..i+t-1] = {\bf x}[j..j+t-1]$ and
${\bf x}[i+t] < {\bf x}[j+t]$.
\end{proof}

\begin{theorem}
Let ${\bf x}$ be a $k$-automatic sequence.
The predicate $L_\infty(i)$ defined by
``${\bf x}[i..\infty]$ is an infinite Lyndon word''
is expressible.
\end{theorem}

\begin{proof}
This is equivalent to $LL_\infty(i,j)$ holding for all $j > i$.
\end{proof}

\section{Lyndon factorization}

Siromoney et al.~\cite{Siromoney} proved that
every infinite word ${\bf x} = a_0 a_1 a_2 \cdots$
can be factorized uniquely in exactly one of the following
two ways:

\begin{itemize}
\item[(a)] as ${\bf x} = w_1 w_2 w_3 \cdots$ where each $w_i$ is a
finite Lyndon word
and $w_1 \geq w_2 \geq w_3 \cdots$;
or 
\item[(b)] as ${\bf x} = w_1 w_2 w_3 \cdots w_r {\bf w}$ where
$w_i$ is a finite Lyndon word for $1 \leq i \leq r$, and ${\bf w}$ is
an infinite Lyndon word, and $w_1 \geq w_2 \geq \cdots \geq w_r \geq {\bf w}$.
\end{itemize}

If (a) holds we say that the Lyndon factorization of $\bf x$ is
infinite; otherwise we say it is finite.

Ido and Melan\c{c}on \cite{Melancon,Ido}
gave an explicit description of the Lyndon factorization
of the Thue-Morse word $\bf t$
and the period-doubling sequence (among other things).
(Recall that the Thue-Morse word is given by
${\bf t}[n] = $ the number of $1$'s in the binary expansion of
$n$, taken modulo $2$.)
For the Thue-Morse word, this factorization is given by
$${\bf t} = w_1 w_2 w_3 w_4 \cdots = (011) (01) (0011) (00101101) \cdots,$$
where each term in the factorization, after the first, is double the
length of the previous.  S\'e\'ebold \cite{Seebold:2003} and
\v{C}ern\'y generalized these results to other related automatic
sequences.

In this section, generalizing the work of Ido, Melan\c{c}on,
S\'e\'ebold, and \v{C}ern\'y,
we prove that the Lyndon factorization of a $k$-automatic sequence
is itself $k$-automatic.  Of course, we need to explain how the
factorization is encoded.  The easiest and most natural way to do this
is to use an infinite word over $\lbrace 0,1 \rbrace$, where the $1$'s
indicate the positions where a new term in the factorization begins.
Thus the $i$'th $1$, for $i \geq 0$,
appears at index $|w_1 w_2 \cdots w_i|$.
For example, for the Thue-Morse word, this encoding is given by
$$100101000100000001 \cdots .$$
If the factorization is infinite, then there are infinitely many
$1$'s in its encoding; otherwise there are finitely many $1$'s.

In order to prove the theorem, we need a number of results.
We draw a distinction between a {\it factor} $f$ of $\bf x$ (which is just
a word) and an {\it occurrence} of
that factor (which specifies the exact position at which $f$ occurs).
For example, in the Thue-Morse word $\bf t$, the
factor $0110$ occurs as ${\bf x}[0..3]$ and ${\bf x}[11..15]$ and many
other places.  We call $[0..3]$ and $[11..15]$, and so forth, the
{\it occurrences} of $0110$.  An occurrence is said to be Lyndon if the
word at that position is Lyndon.    We say an occurrence $O_1 = [i..j]$ is
{\it inside} an occurrence $O_2 = [i'..j']$ if $i' \leq i$ and $j' \geq j$.
If, in addition, either $i' < i$ or $j < j'$ (or both), then we say
$O_1$ is {\it strictly inside} $O_2$.  These definitions are easily extended
to the case where $j$ or $j'$ are equal to $\infty$, and they
correspond to the predicates $I$ (inside) and $SI$ (strictly inside)
given below:
\begin{eqnarray*}
I(i,j,i',j') & \text{is} & \quad i' \leq i \text{ and } j' \geq j \\
SI(i,j,i',j') & \text{is} & \quad I(i,j,i',j') \text{ and } 
	((i' < i) \text{ or } (j' > j)) 
\end{eqnarray*}

An infinite Lyndon factorization 
$$ {\bf x} = w_1 w_2 w_3 \cdots $$
then corresponds to an infinite sequence of occurrences
$$[i_1..j_1], [i_2..j_2], \cdots$$
where $w_n = {\bf x}[i_n..j_n]$ and $i_{n+1} = j_n + 1$ for $n \geq 1$,
while a finite Lyndon factorization
$$ {\bf x} = w_1 w_2 \cdots w_r {\bf w}$$
corresponds to a finite sequence of occurrences
$$ [i_1..j_1], [i_2..j_2], \ldots, [i_r..j_r], [i_{r+1}..\infty]$$
where $w_n = {\bf x}[i_n..j_n]$ and $i_{n+1} = j_n + 1$
for $1 \leq n \leq r$.

\begin{theorem}
Let $\bf x$ be an infinite word.  Every Lyndon occurrence in $\bf x$ appears
inside a term of the Lyndon factorization of $\bf x$.
\label{lyno}
\end{theorem}

\begin{proof}
We prove the result for infinite Lyndon factorizations; the result for
finite factorizations is exactly analogous.

Suppose the factorization is ${\bf x} = w_1 w_2 w_3 \cdots $.
It suffices to show that no Lyndon occurrence can span the boundary
between two terms of the factorization.  Suppose, contrary to what we
want to prove, that $u w_i w_{i+1} \cdots  w_j v$ is a Lyndon word
for some $u$ that is a nonempty suffix of $w_{i-1}$ (possibly equal to
$w_{i-1}$), and $v$ that is a nonempty prefix of
$w_{j+1}$ (possibly equal to $w_{j+1}$), and
and $i \leq j+1$.  (If $i = j+1$ then there are no $w_i$'s at
all between $u$ and $v$.)  

Since $u$ is a suffix of $w_{i-1}$ and $w_{i-1}$ is Lyndon, we have
$u \geq w_{i-1}$.
On the other hand, by the Lyndon factorization definition
we have $w_{i-1} \geq w_i \geq \cdots \geq w_j \geq  w_{j+1}$.
But $v$ is a prefix of $w_{j+1}$, so just by the definition of lexicographic
ordering we have $w_{j+1} \geq v$.
Putting this all together we get $u \geq v$.  So $ux \geq v$ for all 
words $x$.

On the other hand, since $u w_i \cdots w_j v$ is Lyndon, it must be
lexicographically less than any proper suffix --- for instance,
$v$.  So $u w_i
\cdots w_j v < v$.    Take $x = w_i \cdots w_j v$ to get a contradiction
with the conclusion in the previous paragraph.
\end{proof}

\begin{corollary}
The occurrence $[i..j]$ corresponds to a term in the Lyndon factorization
of $\bf x$ if and only if
\begin{itemize}
\item[(a)] $[i..j]$ is Lyndon; and
\item[(b)] $[i..j]$ does not occur strictly inside any other Lyndon
occurrence.
\end{itemize}
\label{eight}
\end{corollary}

\begin{proof}
Suppose $[i..j]$ corresponds to a term $w_n$ in the Lyndon factorization
of $\bf x$.  Then evidently $[i..j]$ is Lyndon.  If it occurred
strictly inside some other Lyndon occurrence, say $[i'..j']$,
then we know from Theorem~\ref{lyno} that $[i'..j']$ itself lies in 
inside some $w_m$, so $[i..j]$ must lie strictly inside $w_m$,
which is clearly impossible.

Now suppose $[i..j]$ is Lyndon and does not occur strictly inside any
other Lyndon occurrence.  
From Theorem~\ref{lyno} $[i..j]$ must occur inside
some term of the factorization $[i'..j']$.  If $[i..j] \not= [i'..j']$ then 
$[i..j]$ lies strictly inside $[i'..j']$, a contradiction.  So
$[i..j] = [i'..j']$ and hence corresponds to a term of the
factorization.
\end{proof}

\begin{corollary}
The predicate $LF(i,j)$ defined by ``$[i..j]$ corresponds to a term
of the Lyndon factorization of $\bf x$'' is expressible.  
\end{corollary}

\begin{proof}
Indeed, by Corollary~\ref{eight}, the predicate $LF(i,j)$ can be
defined by
$$ L(i,j) \text{ and } \forall \ i', j' \ ( SI(i,j,i',j') \implies \neg L(i',j') ).$$
\end{proof}

We can now prove the main result of this section.

\begin{theorem}
Using the encoding mentioned above,
the Lyndon factorization of a $k$-automatic sequence
is itself $k$-automatic.
\label{lf}
\end{theorem}

\begin{proof}
Using the technique of \cite{Allouche&Rampersad&Shallit:2009},
we can create an automaton that on input
$i$ expressed in base $k$, guesses $j$ and checks if $LF(i,j)$ holds.
If so, it outputs $1$ and otherwise $0$.    To get the last $i$ in the
case that the Lyndon factorization is finite, we also accept $i$
if $L_\infty(i)$ holds.
\end{proof}

We also have

\begin{theorem}
Let $\bf x$ be a $k$-automatic sequence.
It is decidable if the Lyndon factorization of $\bf x$ is finite or
infinite.
\end{theorem}

\begin{proof}
The construction given above in the proof of Theorem~\ref{lf} produces
an automaton that accepts finitely many distinct $i$ (expressed in 
base $k$) if and only if the Lyndon factorization of $\bf x$ is finite.
\end{proof}

We programmed up our method and found the Lyndon factorization of
the Thue-Morse sequence $\bf t$, the period-doubling sequence $\bf d$,
the paperfolding
sequence $\bf p$, and the Rudin-Shapiro sequence $\bf r$, and their
negations. (The results for
Thue-Morse and the period-doubling sequence were already given
in \cite{Ido}, albeit in a different form.)
Recall that the period-doubling
sequence is defined by ${\bf p}[n] = | {\bf t}[n+1] - {\bf t}[n] |$.
The paperfolding sequence ${\bf p} = 0010011 \cdots$
arises from the limit of the sequence 
$(f_n)$, where $f_0 = 0$ and $f_{n+1} = f_n 0 \overline{f_n}^R$, where
$R$ denotes reversal and $\overline{x}$ maps $0$ to $1$ and $1$ to $0$.
Finally, the Rudin-Shapiro sequence $\bf r$ is defined by
${\bf r}[n] = $ the number of (possibly overlapping) occurrences of
$11$ in the binary expansion of $n$, taken modulo $2$.
The results are given in the theorem below.

\begin{theorem}
The occurrences corresponding to the Lyndon factorization of
each word is as follows:
\begin{itemize}
\item the Thue-Morse sequence $\bf t$:
$[0..2], [3..4], [5..8], [9..16], [17..32],
	\ldots, [2^i + 1..2^{i+1}], \ldots $;
\item the negated Thue-Morse sequence $\overline{\bf t}$:
$[0..0], [1..\infty]$;
\item the Rudin-Shapiro sequence $\bf r$:  $[0..6],[7..14],[15..30],
	\ldots, [2^i - 1..2^{i+1} - 2], \ldots$;
\item the negated Rudin-Shapiro sequence $\overline{\bf r}$:
$[0..0],[1..1],[2..2],[3..10],[11..42],[43..46],[47..174],\ldots,
[4^i-4^{i-1}-4^{i-2}-1..4^i-4^{i-1}-2],
[4^i-4^{i-1}-1..4^{i+1}-4^i-4^{i-1}-1], \ldots$;
\item the paperfolding sequence $\bf p$:
$[0..6],[7..14],[15..30], \ldots, [2^i - 1..2^{i+1}-2], \ldots$;
\item the negated paperfolding sequence $\overline{\bf p}$:  
	$[0..0], [1..1], [2..4], [5..9], [10..20], [21..84], [85..340], \ldots,
		[(4^i - 1)/3..4(4^i - 1)/3], \ldots $;
\item the period-doubling sequence $\bf d$:
	$[0..0], [1..4], [5..20], [21..84], \ldots,
		[(4^i - 1)/3..4(4^i - 1)/3], \ldots$;
\item the negated period-doubling sequence $\overline{\bf d}$:
	$[0..1],[2..9],[10..41],[42..169], \ldots,
		 [2(4^i -1)/3..2(4^{i+1}-1)/3-1], \ldots$.
\end{itemize}
\end{theorem}

\section{Enumeration}

There is a useful generalization of $k$-automatic sequences to sequences
over $\Enn$, the non-negative integers.  A sequence $(a_n)_{n \geq 0}$
over $\Enn$ is called $k$-regular
if there exist vectors $u$ and $v$ and a matrix-valued morphism
$\mu$ such that $a_n = u \mu(w) v$, where $w$ is the base-$k$ representation
of $n$.  For more details, see \cite{Allouche&Shallit:1992}.

The subword complexity function $\rho(n)$ of an infinite sequence $\bf x$
counts the number of distinct length-$n$ factors of $\bf x$.
There are also many variations, such as counting the number of
palindromic factors or unbordered factors.    If $\bf x$ is $k$-automatic,
then all three of these are $k$-regular sequences
\cite{Allouche&Rampersad&Shallit:2009}.  We now show that the same
result holds for the number $\rho_{\bf x}^P (n)$ of primitive factors and
for the number $\rho_{\bf x}^L$ of Lyndon factors.  We refer to these
two quantities as the ``primitive complexity'' and ``Lyndon complexity'',
respectively.

\begin{theorem}
The function counting the number of length-$n$ primitive (resp., Lyndon)
factors of a $k$-automatic sequence $\bf x$ is $k$-regular.
\label{primlyn}
\end{theorem}

\begin{proof}
By the results of \cite{Charlier&Rampersad&Shallit:2011},
it suffices to show that there is an automaton accepting
the base-$k$ representations of pairs $(n,i)$ such that the number
of $i$'s associated with each $n$ equals the number of primitive
(resp., Lyndon) factors of length $n$.

To do so, it suffices to show that the predicate
$P(n,i)$ defined by ``the factor of length $n$ beginning at position $i$
is primitive (resp., Lyndon) and is the first occurrence of that factor
in $\bf x$'' is expressible.  This is just
$$P(i,i+n-1) \quad \text{and} \quad \forall j < i \ {\bf x}[i..i+n-1] \not= {\bf x}[j..j+n-1],$$
(resp.,
$$L(i,i+n-1) \quad \text{and} \quad  \forall j < i \ {\bf x}[i..i+n-1] \not= {\bf x}[j..j+n-1]).$$
\end{proof}

We used our method to compute these sequences for the Thue-Morse sequence,
and the results are given below.  

\begin{theorem}
Let $\rho_{\bf t}^L (n)$ denote the number of Lyndon factors of length
$n$ of the Thue-Morse sequence.  Then
$$\rho_{\bf t}^L (n) = 
\begin{cases} 
1, & \text{if $n = 2^k$ or $5 \cdot 2^k$ for $k \geq 1$ } ; \\
2, & \text{if $n = 1$ or $n = 5$ or $n = 3 \cdot 2^k$ for $k \geq 0$} ; \\
0, & \text{otherwise.}
\end{cases}
$$
\label{lyndon-tm}
\end{theorem}

\begin{theorem}
Let $\rho_{\bf t}^P (n)$ denote the number of primitive factors of length
$n$ of the Thue-Morse sequence.  Then
$$
\rho_{\bf t}^P (n) =
\begin{cases}
3\cdot 2^t - 4, & \text{if $n = 2^t$;} \\
4n-2^t - 4, & \text{if $2^t + 1 \leq n < 3 \cdot 2^{t-1}$}; \\
5 \cdot 2^t - 6, & \text{if $n = 3 \cdot 2^{t-1}$}; \\
2n + 2^{t+1} - 2, & \text{if $3 \cdot 2^{t-1} < n < 2^{t+1}$}.
\end{cases}
$$
\label{prim-tm}
\end{theorem}

We can also state a similar result for the Rudin-Shapiro sequence.

\begin{theorem}
Let $\rho_{\bf r}^L (n)$ denote the Lyndon complexity
of the Rudin-Shapiro sequence.  Then
$ \rho_{\bf r}^L (n) \leq 8$ for all $n$.  This sequence is $2$-automatic
and there is an automaton of 2444 states that generates it.
\label{lyndon-rs}
\end{theorem}

% APL workspace GOC3; function CREATE

\begin{proof}
The proof was carried out by machine computation, and we briefly
summarize how it was done.

First, we created an automaton $A$ to accept all pairs of integers $(n,i)$,
represented in base $2$, such that the factor of length $n$ in
$\bf r$, starting at position $i$, is a Lyndon factor, and is the first
occurrence of that factor in $\bf r$.  Thus, the number of distinct
integers $i$ associated with each $n$ is $\rho_{\bf r}^L (n)$.  
The automaton $A$ has $102$ states.

Using the techniques in \cite{Charlier&Rampersad&Shallit:2011}, we then
used $A$ to create matrices $M_0$ and $M_1$ of dimension
$102 \times 102$, and vectors $v, w$
such that $v M_x w = \rho_{\bf r}^L (n)$, if $x$ is the base-$2$
representation of $n$.  Here if $x = a_1 a_2 \cdots a_i$, then
by $M_x$ we mean the product $M_{a_1} M_{a_2} \cdots M_{a_i}$.

From this we then created a new automaton $A'$ where the states are
products of the form $v M_x$ for binary strings $x$ and the
transitions are on $0$ and $1$.  This automaton was
built using a breadth-first approach, using a queue to hold states whose
targets on $0$ and $1$ are not yet known.  From Theorem~\ref{lyndon-bounded} 
in the next
section, we know that $\rho_{\bf r}^L (n)$ is bounded, so
that this approach must terminate.
It did so at
2444 states, and the product of the $v M_x$ corresponding to each state
with $w$ gives an integer less than or equal to $8$, thus proving the
desired result and also providing an automaton to compute $\rho_{\bf r}^L
(n)$.
\end{proof}

\begin{remark}
Note that the Lyndon complexity functions in Theorems~\ref{lyndon-tm} and
\ref{lyndon-rs} are bounded.  This will follow more generally from
Theorem~\ref{lyndon-bounded} below.
\end{remark}

\section{Finite factorizations}

Of course, the original Lyndon factorization was for finite words:
every finite nonempty word $x$ can be factored uniquely as a 
nonincreasing product $w_1 w_2 \cdots w_m$
of Lyndon words.  We can apply this theorem
to all prefixes of a $k$-automatic sequence.
It is then natural to wonder
if a {\it single\/} automaton can encode {\it all\/} the 
Lyndon factorizations of {\it all\/}
finite prefixes.  The answer is yes, as the following result shows.

\begin{theorem}
Suppose $\bf x$ is a $k$-automatic sequence.  Then there is an
automaton $A$ accepting
\begin{align*}
& \lbrace (n,i)_k \ : \ \text{the Lyndon factorization of 
	${\bf x}[0..n-1]$ is $w_1 w_2 \cdots w_m$} \\
& \quad\quad \text{with} \ w_m = {\bf x}[i..n-1]
	\rbrace.
\end{align*}
\end{theorem}

\begin{proof}
As is well-known \cite{Duval:1983},
if $w_1 w_2 \cdots w_m$ is the Lyndon factorization
of $x$, then $w_m$ is the lexicographically least suffix of $x$.
So to accept $(n,i)_k$ we find $i$ such that 
${\bf x}[i..n-1] < {\bf x}[j..n-1]$ for 
$0 \leq j < n$ and $i \not= j$.
\end{proof}

Given $A$, we can find the complete factorization of any prefix
${\bf x}[0..n-1]$ by using this automaton to
find the appropriate $i$ (as described
in \cite{Goc&Schaeffer&Shallit:2012}) and then replacing
$n$ with $i$.

We carried out this construction for the Thue-Morse sequence, and
the result is shown below in Figure~\ref{fig1}.

\begin{figure}[htb]
\leavevmode
\def\epsfsize#1#2{0.3#1}
\centerline{\epsfbox{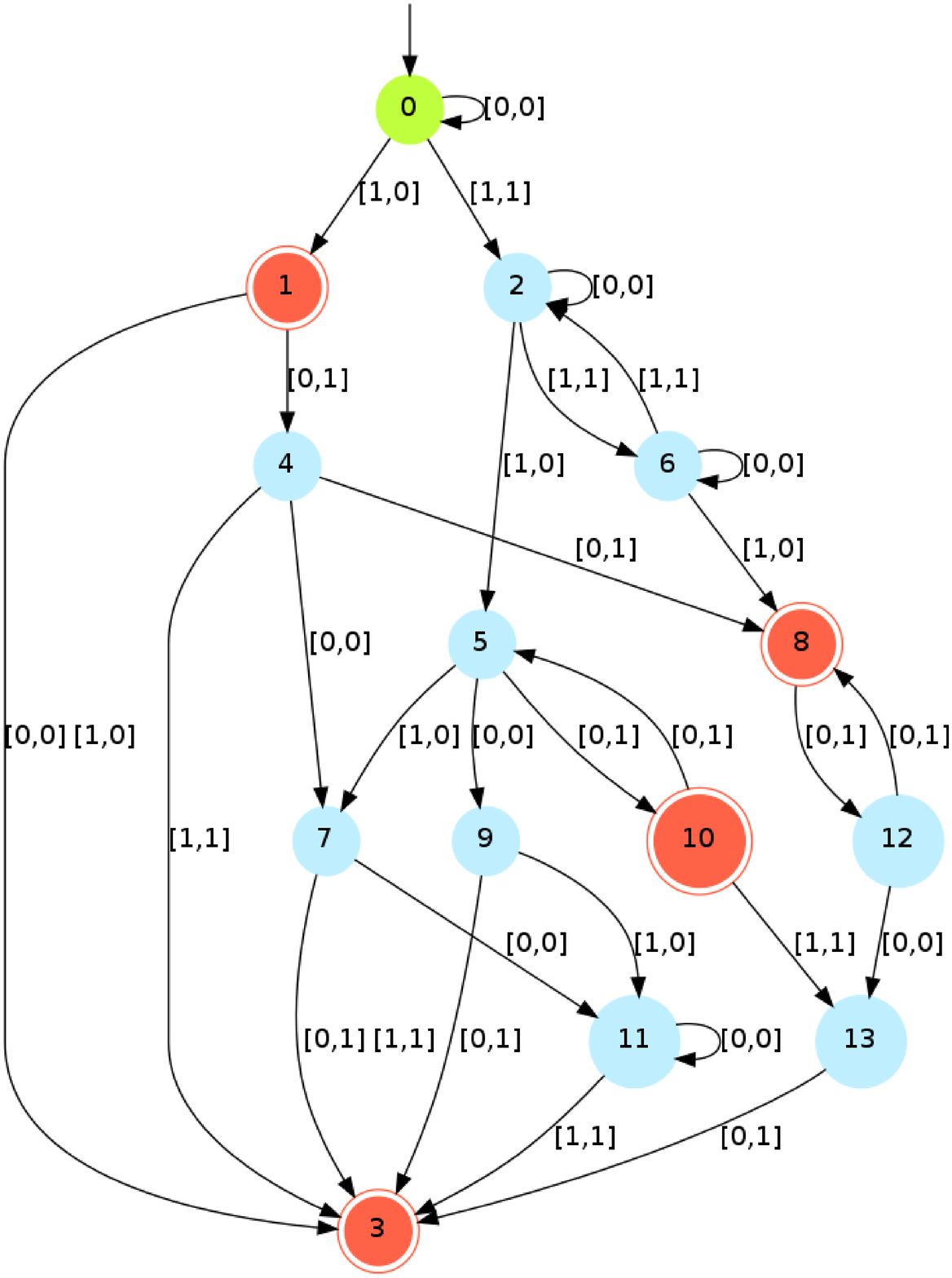}}
\protect\label{fig1}
\caption{A finite automaton accepting the base-$2$ representation of
$(n,i)$ such that the Lyndon factorization of
${\bf t}[0..n-1]$ ends in the term ${\bf t}[i..n-1]$}
\end{figure}

In a similar manner, there is an automaton that encodes the
factorization of {\it every} factor of a $k$-automatic sequence:

\begin{theorem}
Suppose $\bf x$ is a $k$-automatic sequence.  Then there is an
automaton $A'$ accepting
\begin{align*}
& \lbrace (i,j,l)_k \ : \ \text{the Lyndon factorization of 
	${\bf x}[i..j-1]$ is $w_1 w_2 \cdots w_m$ } \\
& \quad\quad \text{with $w_m = {\bf x}[l..n-1]$} 
	\rbrace.
\end{align*}
\end{theorem}

We calculated $A'$ for the Thue-Morse sequence using our
method.  It is a 34-state machine and is
displayed in Figure~\ref{fig2}.

\begin{figure}[H]
\leavevmode
\def\epsfsize#1#2{0.13#1}
\centerline{\epsfbox{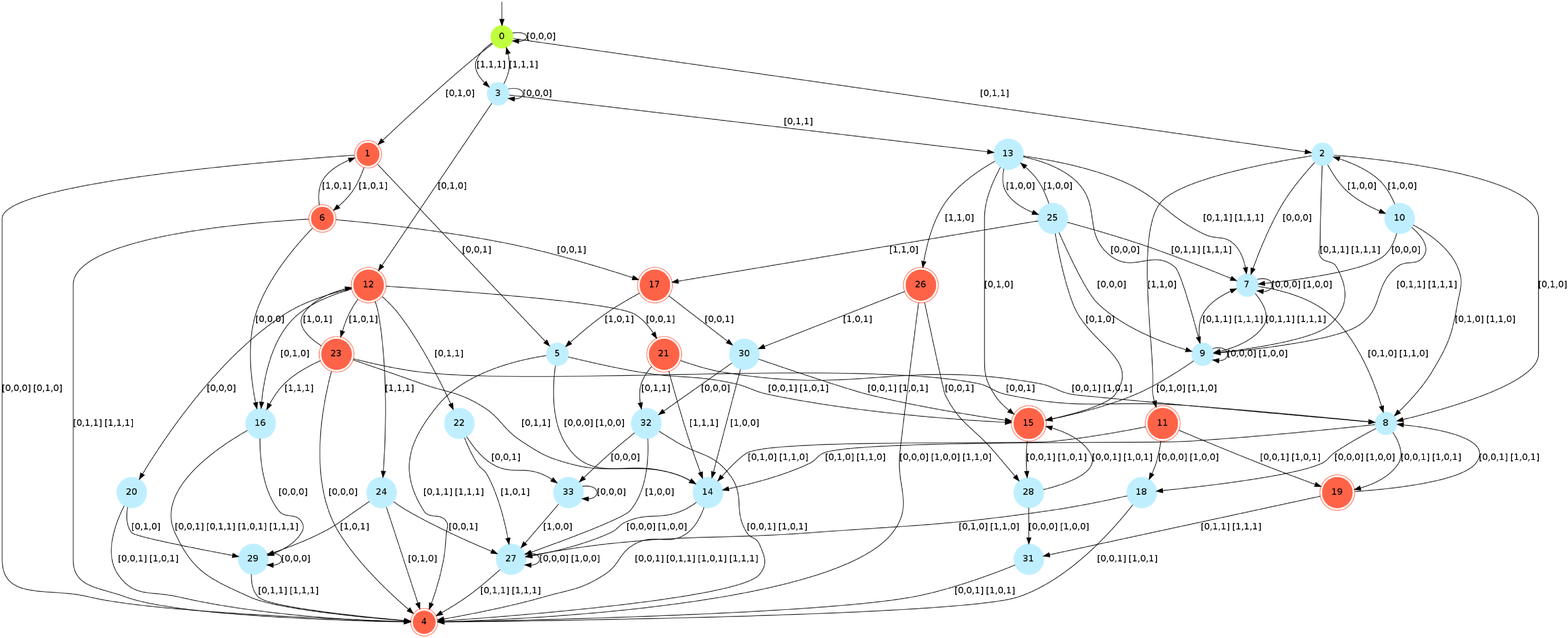}}
\protect\label{fig2}
\caption{A finite automaton accepting the base-$2$ representation of
$(i,j,l)$ such that the Lyndon factorization of
${\bf t}[i..j-1]$ ends in the term ${\bf t}[l..j-1]$}
\end{figure}

Another quantity of interest is the number of terms in the
Lyndon factorization of each prefix.

\begin{theorem}
Let $x$ be a $k$-automatic sequence.  Then the sequence
$(f(n))_{n \geq 0}$ defined by
$$ f(n) = \text{the number of terms in the Lyndon factorization of
	${\bf x}[0..n]$} $$
is $k$-regular.
\end{theorem}

\begin{proof}
We construct an automaton to accept
$$\lbrace (n,i) \ : \ \exists j \leq n \text{ such that $L(i,j)$ and
	if $SI(i,j,i',j')$ and $0 \leq i' \leq j' \leq n$
	then $\neg L(i', j')$} \rbrace.$$
\end{proof}	

For the Thue-Morse sequence the corresponding sequence satisfies
the relations
\begin{eqnarray*}
f(4n+1) &=& -f(2n) + f(2n+1) + f(4n) \\
f(8n+2) &=& -f(2n) + f(4n) + f(4n+2) \\
f(8n+3) &=& -f(2n) + f(4n) + f(4n+3) \\
f(8n+6) &=& -f(2n) - f(4n+2) + 3f(4n+3) \\
f(8n+7) &=& -f(2n) + 2f(4n+3) \\
f(16n) &=& -f(2n) + f(4n) + f(8n) \\
f(16n+4) &=& -f(2n) + f(4n) + f(8n+4) \\
f(16n+8) &=& -f(2n) + f(4n+3) + f(8n+4) \\
f(16n+12) &=& -f(2n) -2f(4n+2) + 3f(4n+3) + f(8n+4) 
\end{eqnarray*}
for $n \geq 1$, which allows efficient calculation of this quantity.

\section{Linearly recurrent sequences}

\begin{definition}
A recurrent infinite word ${\bf x} = a_{0}a_{1} a_{2}  \cdots$, where
each $a_i$ is a letter,
is called \emph{linearly recurrent with constant} $L>0$ if, for every factor $u$
and its two consecutive occurrences beginning at positions
$i$ and $j$ in ${\bf x}$ with $i < j$, we have $j - i < L \nabs{u}$.
The word $a_{i}a_{i+1} \cdots a_{j-1}$ is called a \emph{return word} of $u$.
Thus linear recurrence can be defined from the condition that every return word $w$ of every factor $u$ of ${\bf x}$ satisfy  $\nabs{w} < L \nabs{u}$.
Let ${\cal R}_u$ denote the set of return words of $u$ in ${\bf x}$.
\end{definition}

\begin{remark}\label{su 26-08-2012 00:34}
 Linear recurrence implies that every length-$k$ factor appears at
 least once in every factor of length $(L+1)k - 1$.
\end{remark}

\begin{lemma}[Durand, Host, and Skau~\cite{DurHosSka1999}] \label{su 26-08-2012 00:38}
Let ${\bf x}$ be an aperiodic linearly recurrent word with constant~$L$.
\begin{enumerate}[(i)]
%\item The word ${\bf x}$ is $(L+1)$-power free.
\item If $u$ is a factor of ${\bf x}$ and $w$ its return word, then $\nabs{w} > \nabs{u}/L$.
\item The number of return words of any given factor $u$ of ${\bf x}$ is  $ \# 
{\cal R}_u \leq L(L+1)^{2}$.
\end{enumerate}
\end{lemma}

\begin{theorem}\label{su 26-08-2012 00:48}
The Lyndon complexity of any linearly recurrent sequence is bounded.
\label{lyndon-bounded}
\end{theorem}

\begin{proof}
Let ${\bf x}$ be a linearly recurrent sequence with constant $L$.
If ${\bf x}$ is ultimately periodic, it is purely periodic because it is recurrent, and thus its Lyndon complexity is bounded.
Therefore assume that ${\bf x}$ is aperiodic, and let $n\geq L$.
Denote $k = \nfloor{(n+1)/(L+1)}$, so that
\begin{equation}\label{su 26-08-2012 01:08}
(L+1) k - 1  \leq n < (L+1)(k+1) - 1.
\end{equation}
The left-hand side inequality in~\eqref{su 26-08-2012 01:08} and Remark~\ref{su 26-08-2012 00:34} together imply that
all factors in ${\bf x}$ of length $k$ occur in all factors of length $n$. Therefore if $u$  is the lexicographically smallest factor of length~$k$, then
every Lyndon factor of ${\bf x}$ of length~$n$ must begin with~$u$. Since every suffix of ${\bf x}$ that begins with $u$ can be factorized over 
${\cal R}_u$,
we conclude further that every length-$n$ Lyndon factor of ${\bf x}$
is a prefix of a word in ${\cal R}_u^*$.

The return words of $u$ have length at least $k/L$ by Lemma~\ref{su 26-08-2012 00:38}.
Furthermore, the right-hand side inequality in~\eqref{su 26-08-2012 01:08} gives
\[
\frac{n}{k/L} < \frac{(L+1)(k+1) - 1}{k/L} < \frac{L(L+1)(k+1)}{k} \leq  2L(L+1).
\]
Therefore any Lyndon factor of length $n$ is a prefix of a word in 
${\cal R}_{u}^{2L(L+1)}$.  
Since $\# {\cal R}_{u}\leq L(L+1)^{2}$ by Lemma~\ref{su 26-08-2012 00:38},
we conclude that
\[
\rho_{\bf x}^L(n) \leq \max\bsset{\rho_{\bf x}^L(1), \rho_{\bf x}^L(2), \ldots, \rho_{\bf x}^L(L-1), L(L+1)^{4L(L+1)}},
\]
so that the Lyndon complexity of ${\bf x}$ is bounded.
\end{proof}

\begin{definition}
Let $h \colon {\cal A}^{*} \rightarrow {\cal A}^{*}$ be a primitive morphism, and let $\tau \colon {\cal A} \rightarrow {\cal B}$
be a letter-to-letter morphism.
If $h$ is prolongable, so that the limit $h^{\omega}(a) := \lim_{\ntoinf}h^{n}(a)$ exists for some letter $a\in {\cal A}$, then the sequence
$\tau\bparen{h^{\omega}(a)}$ is called \emph{primitive morphic}.
\end{definition}

\begin{lemma}[Durand~\cite{Durand1998,DurHosSka1999}]\label{la 25-08-2012 23:36}
Primitive morphic sequences are linearly recurrent.
\end{lemma}

\begin{corollary}\label{to 04-10-2012 11:42}
The Lyndon complexity of any primitive morphic sequence is bounded.
\end{corollary}
\begin{proof}
Follows from Lemma~\ref{la 25-08-2012 23:36} and Theorem~\ref{su 26-08-2012 00:48}.
\end{proof}

\begin{corollary}
If ${\bf x}$ is $k$-automatic and primitive morphic, then its Lyndon complexity is $k$-automatic.
\end{corollary}
\begin{proof}
Follows from Corollary~\ref{to 04-10-2012 11:42} and
Theorem~\ref{primlyn}, because a $k$-regular sequence over a finite
alphabet is $k$-automatic \cite{Allouche&Shallit:1992}.
\end{proof}

%\begin{remark}
% There are uniformly recurrent infinite words with bounded Lyndon complexity, but that are not linearly recurrent. For example, if ${\bf x}$ is a Sturmian word with a slope whose partial coefficients are unbounded, then it is not linearly recurrent~\cite{Durand2000}, but its Lyndon complexity is bounded by~$1$, see~\cite{CurSaa2009}.
%\end{remark}

\section{Acknowledgments}

We thank Luke Schaeffer for suggesting the argument in the proof of
Theorem~\ref{prim-ex}.

\end{document}